\documentclass{article}[12pt]

\usepackage{amsmath}
\usepackage{amsmath,amssymb,amsthm,graphicx}
\usepackage{amsfonts}
\usepackage{pstricks}
\usepackage{verbatim}
\usepackage{bbm}
\usepackage{mathrsfs, hyphenat}
\usepackage[makeroom]{cancel}
\usepackage[hidelinks]{hyperref}  
\usepackage[anythingbreaks]{breakurl}  
\usepackage[all]{nowidow} 
\usepackage{amsthm}

\usepackage[caption=false, position=top, singlelinecheck=false]{subfig} 
\newtheorem{theorem}{Theorem}
\newtheorem{proposition}{Proposition}

\newtheorem{defn}{Definition}[section]
\title{Transposed Poisson superalgebra}

\date{}
\author{Viktor Abramov \\ {\small Institute of Mathematics and Statistics},\\
        {\small University of Tartu, Estonia}
\and
Olga Liivapuu \\ {\small Institute of Forestry and Engineering,}\\
                  {\small Estonian University of Life Sciences, Estonia}}

\begin{document}

\maketitle
\begin{abstract}
{In this paper we propose the notion of a transposed Poisson superalgebra. We prove that a transposed Poisson superalgebra can be constructed by means of a commutative associative superalgebra and an even degree derivation of this algebra. Making use of this we construct two examples of transposed Poisson superalgebra. One of them is the graded differential algebra of differential forms on a smooth finite dimensional manifold, where we use the Lie derivative as an even degree derivation. The second example is the commutative superalgebra of basic fields of the quantum Yang-Mills theory, where we use the BRST-supersymmetry as an even degree derivation to define a graded Lie bracket. We show that a transposed Poisson superalgebra has six identities that play an important role in the study of the structure of this algebra.}
\end{abstract}


\section{INTRODUCTION}
Poisson algebras play an important role in many branches of differential geometry, theoretical mechanics, and field theory. Thus this area of research is actively developing, as evidenced by various generalizations of the concept of Poisson algebra. One direction of development in the theory of Poisson algebras is the extension of the concept of Poisson algebra to structures with an $n$-ary multiplication law. Initially such a generalization was proposed by Nambu in his paper \cite{Nambu}  and the model of quarks in the theory of elementary particles served as a stimulus for such a generalization. Later, this generalization called Nambu-Poisson algebra was studied in a number of papers, and an excellent survey of this direction in the development of the theory of Poisson algebras is the paper by Takhtajan \cite{Takhtajan}.  A Nambu-Poisson algebra is a $n$-Lie algebra, where $n$-Lie algebra is based on an extension of binary Lie bracket to $n$-ary Lie bracket. The concept of a $n$-Lie algebra was proposed by Filippov \cite{Filippov}. The elements of a $n$-Lie algebra satisfy the Filippov-Jacobi identity, which is a generalization of Jacobi identity to $n$-ary bracket. The classification of simple linearly compact $n$-Lie superalgebras is given in \cite{Cantarini-Kac}. A class of ternary Lie superalgebras constructed by means of the supertrace and applications of these ternary Lie superalgebras in BRST-supersymmetries are proposed in  \cite{Abramov 1, Abramov 2}.

Recently, the notion of a transposed Poisson algebra has been proposed. A transposed Poisson algebra can be considered as a structure dual to the concept of a Poisson algebra. It was shown in \cite{Bai-Guo-Wu-2020} that a transposed Poisson algebra in its properties is very similar to a Poisson algebra. For example, the class of transposed Poisson algebras is closed under the tensor product of such algebras. In \cite{Bai-Guo-Wu-2020} it is also indicated that there is an important connection between transposed Poisson algebras and Novikov-Poisson algebras. More exactly, if we equip a Novikov-Poisson algebra with the Lie commutator constructed with the help of multiplication in Novikov-Poisson algebra, then the Novikov-Poisson algebra becomes a transposed Poisson algebra. In \cite{Bai-Guo-Wu-2020} it is shown that a transposed Poisson algebra possesses a number of identities and this is an important property of a transposed Poisson algebra. In \cite{Kaygorodov} all structures of complex transposed Poisson algebras on Galilean type Lie algebras and superalgebras are described.

In this paper we propose the notion of a transposed Poisson superalgebra. In analogy with transposed Poisson algebra a transposed Poisson superalgebra is a vector space endowed with two structures, where one of them is a commutative associative superalgebra and the second is a Lie superalgebra. These two structures satisfy the graded compatibility condition. We prove that if $\mathfrak G$ is a commutative associative superalgebra and $\delta$ is an even derivation of $\frak G$ then the graded Lie bracket
\begin{equation}
[u,v]_\delta=u\cdot\delta(v)-(-1)^{|u|\,|v|}\,v\cdot\delta(u),
\label{introduction 1}
\end{equation}
where $u,v\in{\mathfrak G}$, $|u|,|v|$ are degrees of elements $u,v$ respectively, defines the structure of a transposed Poisson superalgebra on $\mathfrak G$. We propose the classification of transposed Poisson superalgebras in low-dimensional case, that is, in the case of (1,1)-superalgebras. We give two examples of transposed Poisson superalgebra based on the graded Lie bracket (\ref{introduction 1}). The first example is a transposed Poisson superalgebra constructed with the help of graded differential algebra. Particularly, we show that the algebra of differential forms on a smooth finite-dimensional manifold $M^n$ can be equipped with a structure of a transposed Poisson superalgebra if we use (\ref{introduction 1}), where $u,v$ are differential forms and $\delta$ is the Lie derivative $\cal L_{\mathtt X}$, $\mathtt X$ is a vector field on $M^n$. We prove that the graded Lie bracket of two differential forms induces a graded Lie bracket on the cohomology classes of differential forms. The second example of a transposed Poisson superalgebra is related to quantum Yang-Mills theory. The BRST-supersymmetry of the quantum Yang-Mills theory can be considered as an even derivation of the superalgebra of basic fields of theory. We construct the graded Lie bracket on this superalgebra with the help of operator of BRST-supersymmetry. Finally we prove that a transposed Poisson superalgebra possesses a rich class of identities which can be considered as graded versions of identities proposed in \cite{Bai-Guo-Wu-2020}.
\section{Transposed Poisson Superalgebra}
In this section we propose a notion of a graded transposed Poisson algebra. We construct an example of a graded transposed Poisson algebra by means of a graded differential algebra.

Let $\mathbb K$ be a field either of real or complex numbers. A transposed Poisson algebra $\mathfrak{P}$ \cite{Bai-Guo-Wu-2020}  is a vector space over $\mathbb K$ with two algebraic operations. One of them will be denoted by $(x,y)\in {\mathfrak P}\times{\mathfrak P}\to x\cdot y\in{\mathfrak P}$ and it defines a structure of associative commutative algebra on $\mathfrak P$. The second is a Lie bracket $(x,y)\in {\mathfrak P}\times{\mathfrak P}\to [x,y]\in{\mathfrak P}$, that is, bilinear, skew-symmetric mapping which satisfies the Jacoby identity. These two structures, i.e. an associative commutative algebra and a Lie algebra, define a structure of transposed Poisson algebra if they are compatible and the condition of compatibility of these two structures has the following form
\begin{equation}
    2z\cdot[x,y]=[z\cdot x,y]+[x,z\cdot y],\quad x,y,z\in{\mathfrak P}.
    \label{condition of compatibility classic}
\end{equation}
An important example of a transposed Poisson algebra \cite{Bai-Guo-Wu-2020} is an associative commutative algebra $\cal A$ endowed with the Lie bracket $[u,v]_\delta=u\cdot\delta(v)-v\cdot\delta(u)$, where $u,v\in{\cal A}$, $(u,v)\to u\cdot v$ is a multiplication in $\cal A$ and $\delta:{\cal A}\to{\cal A}$ is a derivation. Particularly if $M^n$ is a smooth $n$-dimensional manifold, ${\cal F}(M^n)$ is the algebra of smooth functions on $M^n$ and $\tt X$ is a vector field then
\begin{equation}
[f,g]_{\tt X}=f\,{\tt X}(g)-g\,{\tt X}(f),\;\;\; f,g\in {\cal F}(M^n),
\label{bracket of functions}
\end{equation}
defines transposed Poisson algebra of smooth functions on a manifold $M^n$.

We extend the notion of transposed Poisson algebra to a super case by the following definition.
\begin{defn}
    Let ${\mathfrak G}={\mathfrak G}_0\oplus{\mathfrak G}_1$ be a super vector space over $\mathbb K$, where the degree of a homogeneous element $x\in {\mathfrak G}$ will be denoted by $|x|$, that is, $|x|\in{\mathbb Z}_2$. A transposed Poisson superalgebra is a triple $(\mathfrak G,\,\cdot\,,[\,,\,])$, where $(\mathfrak G,\,\cdot\,)$ is an associative commutative superalgebra and $(\mathfrak G,[\,.\,])$ is a Lie superalgebra. Hence we have the following properties
$$
|x\cdot y|=|x|+|y|,\;\;x\cdot y=(-1)^{|x||y|}\,y\cdot x,
$$
and
\begin{eqnarray}
&& |[x,y]|=|x|+|y|,\nonumber\\
&& [x,y]=-(-1)^{|x|\,|y|}[y,x],\nonumber\\
&& [x,[y,z]]=[[x,y],z]+(-1)^{|x|\,|y|}[y,[x,z]].\nonumber
\end{eqnarray}
The compatibility condition for these two structures has the form
\begin{equation}
2z\cdot[x,y]=[z\cdot x,y]+(-1)^{|x|\,|z|}[x,z\cdot y].
\label{condition of compatibility}
\end{equation}
\end{defn}
Obviously $({\mathfrak G}_0,\cdot)$ is an associative commutative algebra and $({\mathfrak G}_0,[\,,\,])$ is a Lie algebra. In the case of even degree elements, i.e. elements of $\mathfrak G_0$, the compatibility condition (\ref{condition of compatibility}) takes on the form (\ref{condition of compatibility classic}). Hence $({\mathfrak G}_0,\cdot,[\,,\,])$ is a transposed Poisson algebra.

An important example of a transposed Poisson algebra is an algebra constructed by means of an associative commutative algebra and its derivation. The next theorem shows that this construction can be extended to the case of superalgebras.
\begin{theorem} Let ${\cal A}={\cal A}_0\oplus{\cal A}_1$ be an associative commutative superalgebra and $\delta:{\cal A}\to{\cal A}$ be an even degree derivation of $\cal A$. Define bracket by
\begin{equation}
[u,v]_\delta=u\cdot\delta(v)-(-1)^{|u|\,|v|}v\cdot\delta(u).
\label{graded commutator}
\end{equation}
Then (\ref{graded commutator}) defines the structure of Lie superalgebra on $\cal A$ and this structure satisfies the compatibility condition (\ref{condition of compatibility}), i.e. the bracket (\ref{graded commutator}) defines the structure of transposed Poisson superalgebra on $\cal A$.
\label{theorem 1}
\end{theorem}
\begin{proof}
It is easy to see that the bracket (\ref{graded commutator}) is graded skew-symmetric, i.e. $[u,v]_\delta=-(-1)^{|u|\,|v|}[v,u]_\delta$ and $|[u,v]|=|u|+|v|$. Then the double brackets can be expressed as follows
\begin{eqnarray}
\big{[}w,[u,v]_\delta\big{]}_\delta &=& w\,u\,\delta^2(v)-\delta(w)\,u\,\delta(v)-w\,\delta^2(u)\,v+\delta(w)\,\delta(u)\,v,\nonumber\\
\big{[}[w,u]_\delta,v\big{]}_\delta &=& -w\,\delta^2(u)\,v+w\,\delta(u)\,\delta(v)+\delta^2(w)\,u\,v-\delta(w)\,u\,\delta(v),\nonumber\\
(-1)^{|u|\,|w|}\big{[}u,[w,v]_\delta\big{]}_\delta &=& w\,u\,\delta^2(v)-w\,\delta(u)\,\delta(v)-\delta^2(w)\,u\,v+\delta(w)\delta(u)\,v.\nonumber
\end{eqnarray}
From this it follows that the graded bracket (\ref{graded commutator}) satisfies the graded Jacobi identity
$$
[w,[u,v]_\delta]_\delta=[[w,u]_\delta,v]_\delta+(-1)^{|u|\,|w|}[u,[w,v]_\delta]_\delta.
$$
Hence the graded bracket (\ref{graded commutator}) defines the Lie superalgebra on commutative superalgebra $\cal A$. We check the compatibility condition (\ref{condition of compatibility}) by straightforward calculations. Indeed we find
\begin{eqnarray}
2\,w\cdot[u,v]_\delta &=& 2\,w\cdot u\cdot\delta(v)-2\,w\cdot \delta(u)\cdot v,\nonumber\\
{[}w\cdot u,v]_\delta &=& w\cdot u\cdot \delta(v)-\delta(w)\cdot u\cdot v - w\cdot\delta(u)\cdot v,\nonumber\\
(-1)^{|u|\,|w|}[u,w\cdot v]_\delta &=& -w\cdot\delta(u)\cdot v+\delta(w)\cdot u\cdot v+w\cdot u\cdot\delta(v).\nonumber
\end{eqnarray}
Thus a commutative superalgebra $\cal A$ endowed with the graded Lie bracket (\ref{graded commutator}) is a transposed Poisson superalgebra.
\end{proof}
\section{Classification of Transposed Poisson Superalgebras in Dimension (1,1)}
In this section we consider a low-dimensional case of a transposed Poisson superalgebra over the field of complex numbers and give a classification of transposed Poisson superalgebras in this case. We consider transposed Poisson superalgebras $\mathfrak G=\mathfrak G_0\oplus\mathfrak G_1$, where $\mbox{dim}\,\mathfrak G_0=\mbox{dim}\,\mathfrak G_1=1$. Let ${ e}_0,{e}_1$ be a basis for a super vector space $\mathfrak G$ such that $e_0$ is an even degree element which spans the even subspace $\mathfrak G_0$ and $e_1$ is an odd degree element which spans the odd subspace $\mathfrak G_1.$ We assume that there is a commutative associative multiplication on a super vector space $\mathfrak G$, which will be denoted by $(x,y)\to x\cdot y, x,y\in\mathfrak G$. From the structure of superalgebra and commutativity of a multiplication it follows
\begin{equation}
e_0\cdot e_0=a\,e_0,\;\; e_0\cdot e_1=e_1\cdot e_0=b\,e_1,\;\;e_1\cdot e_1=0,
\end{equation}
where $a,b$ are complex numbers. In this simple case there is only one combination of three generators which leads to a non-trivial condition derived from associativity. This combination is $e_0,e_0,e_1$. We have
\begin{eqnarray}
e_0\cdot (e_0\cdot e_1)=(e_0\cdot e_0)\cdot e_1\;\Rightarrow\;\;b\;e_0\cdot e_1=a\,e_0\cdot e_1\;\Rightarrow\;b^2\,e_1=ab\,e_1.
\end{eqnarray}
Hence a multiplication is associative in three cases $a=b=0$, $a=b\neq 0$ and $a\neq 0, b=0$.

Now we consider the second structure of a transposed Poisson superalgebra, that is, a Lie superalgebra. It is known that in the case of Lie superalgebras of (1,1)-type there are three non-isomorphic Lie superalgebras
\begin{itemize}
\item[I)] $[e_0,e_0]=0,\;\; [e_0,e_1]=0,\;\; [e_1,e_1]=0,$
\item[II)] $[e_0,e_0]=0,\;\; [e_0,e_1]=0,\;\; [e_1,e_1]=e_0,$
\item[III)] $[e_0,e_0]=0,\;\; [e_0,e_1]=e_1,\;\; [e_1,e_1]=0.$
\end{itemize}
In the case of the Abelian Lie superalgebra I we have two transposed Poisson superalgebras
\begin{eqnarray}
    && e_0\cdot e_0=e_0,\;e_0\cdot e_1=e_1\cdot e_0=0,\;e_1\cdot e_1=0,\nonumber\\
    && [e_0,e_0]=0,\;[e_0,e_1]=0,\;[e_1,e_1]=0,\nonumber
\end{eqnarray}
and
\begin{eqnarray}
    && e_0\cdot e_0=e_0,\;e_0\cdot e_1=e_1\cdot e_0=e_1,\;e_1\cdot e_1=0,\nonumber\\
    && [e_0,e_0]=0,\;[e_0,e_1]=0,\;[e_1,e_1]=0.\nonumber
\end{eqnarray}
In the case of the Lie superalgebra II we have $a=b=0$. Hence we have only one transposed Poisson superalgebra
\begin{eqnarray}
    && e_0\cdot e_0=e_0,\;e_0\cdot e_1=0\cdot e_0=0,\;e_1\cdot e_1=0,\nonumber\\
    && [e_0,e_0]=0,\;[e_0,e_1]=0,\;[e_1,e_1]=e_0.\nonumber
\end{eqnarray}
In the case of the Lie superalgebra III it follows from the compatibility condition $a=b$ and we get one more transposed Poisson superalgebra
\begin{eqnarray}
    && e_0\cdot e_0=e_0,\;e_0\cdot e_1=e_1\cdot e_0=e_1,\;e_1\cdot e_1=0,\nonumber\\
    && [e_0,e_0]=0,\;[e_0,e_1]=e_1,\;[e_1,e_1]=0.\nonumber
\end{eqnarray}
\section{Infinite Dimensional Transposed Poisson Superalgebras}
In this section we apply Theorem (\ref{theorem 1}) to construct a transposed Poisson superalgebra by means of differential forms on a smooth $n$-dimensional manifold $M^n$.

Let us consider a graded commutative differential algebra $({\mathfrak A},d)$, where ${\mathfrak A}=\oplus_{p\in {\mathbb Z}} {\mathfrak A}^p$, and $d:{\mathfrak A}^i\to {\mathfrak A}^{i+1}$ is a differential of $\mathfrak A$, that is, an anti-derivation of degree 1. Let us denote a multiplication in $\mathfrak A$ by $(u,v)\to u\cdot v$, where $u,v\in{\mathfrak A}$ and the degree of a homogeneous element $u$ by $|u|$. Then for $u\in{\mathfrak A}^p$ we have $|u|=p$ and
$$
|u\cdot v|=|u|+|v|,\;\;u\cdot v=(-1)^{|u|\,|v|}v\cdot u,\;\; d(u\cdot v)=(du)\cdot v+(-1)^{|u|}u\cdot dv.
$$
Assume that $\delta$ is a degree (-1) anti-derivation of $\mathfrak A$, i.e. $\delta:{\mathfrak A}^p\to{\mathfrak A}^{p-1}$ and $\delta(u\cdot v)=(\delta u)\cdot v+(-1)^{|u|}u\cdot \delta v$. Then $D=d\circ\delta+\delta\circ d$ is the 0-degree derivation of $\mathfrak A$.
We can consider $\mathfrak A$ as a commutative superalgebra if we define the parity of a homogeneous element $u$ as $|u|(\mbox{mod}\;2)$. Then $d,\delta$ are odd derivations and $D$ is the even derivation of $\mathfrak A$. Define the bracket by
\begin{equation}
[u,v]_D=u\cdot D(v)-(-1)^{|u|\,|v|}\,v\cdot D(u).
\label{bracket with D}
\end{equation}
Then according to Theorem \ref{theorem 1} the triple $({\mathfrak A},\cdot,[\,,\,]_D)$ is a transposed Poisson superalgebra.

As an example of the structure described above we consider the algebra of differential forms on a smooth $n$-dimensional manifold $M^n$. Let $\Omega(M^n)=\oplus_p \Omega^p(M^n)$ be the graded algebra of differential forms on $M^n$. If we denote by ${\cal F}(M^n)$ the commutative algebra of smooth functions on a manifold $M^n$ then $\Omega^0(M^n)$ is identified with the algebra of functions ${\cal F}(M^n)$, i.e. $\Omega^0(M^n)={\cal F}(M^n)$. The graded algebra of differential forms $\Omega(M^n)$ can be considered as a superalgebra if we define the degree of $p$-form $\omega$ by $|\omega|=p\; (\mbox{mod}\;2)$. Then it follows from the property of the wedge product of two forms $\omega\wedge\theta=(-1)^{pq}\theta\wedge\omega$, where $\omega\in\Omega^p(M^n),\theta\in\Omega^q(M^n)$, that $\Omega(M^n)$
is a commutative superalgebra. Let $\tt X$ be a smooth vector field on a manifold $M^n$. Then the operator of contraction of a differential form with a vector field $i_{\mathtt X}:\Omega^{p}\to\Omega^{p-1}$ and the exterior differential $d:\Omega^{p}\to \Omega^{p+1}$ are odd derivations of the commutative superalgebra of differential forms. Then the Lie derivative
$$
{\cal L}_{\tt X}=d\circ i_{\mathtt X}+i_{\mathtt X}\circ d
$$
is an even degree derivation of the commutative superalgebra of differential forms, that is,  ${\cal L}_{\tt X}:\Omega^p(M^n)\to\Omega^p(M^n)$ and ${\cal L}_{\tt X}(\omega\wedge\theta)={\cal L}_{\tt X}(\omega)\wedge \theta+\omega\wedge {\cal L}_{\tt X}(\theta)$. Then the bracket (\ref{bracket with D}) takes on the form
 \begin{eqnarray}
[\omega,\theta]_{\tt X}=\omega\wedge {\cal L}_{\tt X}(\theta)-(-1)^{|\omega|\,|\theta|}\theta\wedge {\cal L}_{\tt X}(\omega),
\label{graded bracket of forms}
 \end{eqnarray}
and the commutative superalgebra of differential forms $\Omega(M^n)$ endowed with the bracket (\ref{graded bracket of forms}) is a transposed Poisson superalgebra. Particularly if we consider two 0-forms, that is, two functions $f,g\in{\cal F}(M^n)$ then the graded Lie bracket (\ref{graded bracket of forms}) takes on the form
$$
[f,g]_{\tt X}=f\,{\cal L}_{\tt X}(g)-g\,{\cal L}_{\tt X}(f)=f\,{\tt X}(g)-g\,{\tt X}(f).
$$
Hence the graded bracket (\ref{graded bracket of forms}) restricted to the subalgebra of smooth functions ${\cal F}(M^n)$ makes it a transposed Poisson algebra (\ref{bracket of functions}). Let $\Omega^p_{\tt c}(M^n)$ be the vector space of closed $p$-differential forms and $\Omega^p_{\tt e}(M^n)\subset \Omega^p_{\tt c}(M^n)$ be the vector space of exact $p$-differential forms.
\begin{proposition}
If $\omega\in\Omega^p_{\tt c}(M^n), \theta\in\Omega^q_{\tt c}(M^n)$ then $[\omega,\theta]_{\tt X}\in \Omega^{p+q}_{\tt c}(M^n)$, i.e. the graded Lie bracket (\ref{graded bracket of forms}) of two closed forms is a closed form. If $\omega\in \Omega^p_{\tt e}(M^n)$ and  $\theta\in\Omega^q_{\tt c}(M^n)$ then $[\omega,\theta]_{\tt X}\in \Omega^{p+q}_{\tt e}(M^n)$, i.e. if one of two closed differential forms is exact then the graded Lie bracket (\ref{graded bracket of forms}) of these two forms is an exact differential form.
\label{proposition 1}
\end{proposition}
\begin{proof}
Let $\omega\in\Omega^p_{\tt c}(M^n), \theta\in\Omega^q_{\tt c}(M^n)$, that is, $d\omega=d\theta=0$. We have
\begin{eqnarray}
[\omega,\theta]_{\tt X} &=& \omega\wedge{\cal L}_{\tt X}(\theta)-(-1)^{|\omega|\,|\theta|}\theta\wedge{\cal L}_{\tt X}(\omega)\nonumber\\
                     &=& \omega\wedge i_{\tt X}\circ d(\theta)+\omega\wedge d\circ i_{\tt X}(\theta)-(-1)^{|\omega|\,|\theta|}(\theta\wedge i_{\tt X}\circ d(\omega)+\theta\wedge d\circ i_{\tt X}(\omega))\nonumber\\
             &=& \omega\wedge d\circ i_{\tt X}(\theta)-(-1)^{|\omega|\,|\theta|} \theta\wedge d\circ i_{\tt X}(\omega).\nonumber
\end{eqnarray}
Now differentiating the graded Lie bracket we obtain
\begin{eqnarray}
d([\omega,\theta]_{\tt X}) &=& d(\omega)\wedge d\circ i_{\tt X}(\theta)+(-1)^{|\omega|}\omega\wedge d^2\circ i_{\tt X}(\theta)\nonumber\\
    &&-(-1)^{|\omega|\,|\theta|} d(\theta)\wedge d\circ i_{\tt X}(\omega)-(-1)^{(|\omega|+1)\,|\theta|} \theta\wedge d^2\circ i_{\tt X}(\omega)=0\nonumber
\end{eqnarray}
If $\omega$ is an exact form, that is, $\omega\in\Omega^p_{\tt e}(M^n), \omega=d\zeta, \zeta\in \Omega^{p-1}(M^n)$ then making use of $d\circ{\cal L}_{\tt X}={\cal L}_{\tt X}\circ d$ we get
\begin{eqnarray}
[\omega,\theta]_{\tt X} = [d\zeta,\theta]_{\tt X} &=& d\zeta\wedge {\cal L}_{\tt X}(\theta)-(-1)^{|\omega|\,|\theta|}\theta\wedge{\cal L}_{\tt X}(d\zeta)\nonumber\\
   &=& d\zeta\wedge {\cal L}_{\tt X}(\theta)-(-1)^{|\omega|\,|\theta|}\theta\wedge d\circ {\cal L}_{\tt X}(\zeta)\nonumber\\
   &=& d\big(\zeta\wedge{\cal L}_{\tt X}(\theta)-(-1)^{|\zeta|\,|\theta|}\theta\wedge {\cal L}_{\tt X}(\zeta)\big).\nonumber
\end{eqnarray}
\end{proof}
Thus the proved proposition shows that the graded Lie bracket (\ref{graded bracket of forms}) induces the graded Lie bracket on the graded algebra of de Rham cohomologies, which defines the structure of transposed Poisson superalgebra.

Our next example of a transposed Poisson superalgebra is related to the quantum theory of Yang-Mills fields \cite{Slavnov-Faddeev}. The gauge group of the Yang-Mills theory is $\mbox{SU}(2)$. Let ${\mathfrak{su}}(2)$ be the Lie algebra of $\mbox{SU}(2)$, $t_a$ be a basis for ${\mathfrak{su}}(2)$ and $[t_a,t_b]=K^c_{ab}t_c$, i.e. $K^c_{ab}$ are structure constants of ${\mathfrak{su}}(2)$. The basic fields of the quantum Yang-Mills theory are ${\mathfrak{su}}(2)$-valued functions $A_\mu, c, \bar c$ defined on a space-time, where $A_\mu$ ($\mu=0,1,2,3$) are Yang-Mills fields, $c,\bar c$ are ghost fields. Since the basic fields are ${\mathfrak{su}}(2)$-valued functions we can express them in the terms of a basis as follows $A_\mu=A^a_\mu\,t_a,c=c^a\,t_a,\bar c=\bar c^a\,t_a$. We consider the algebra of fields of the quantum Yang-Mills theory as an algebra generated by $A^a_\mu,b^a, c^a,\bar c^a$, their all orders partial derivatives with respect to coordinates of space-time and Grassmann variables $\xi,\eta,\ldots$ which do not depend on a point of space-time. We assume that the gauge fields $A^a_\mu$ and an auxiliary field $b^a$ and all their partial derivatives are commuting generators of algebra and they also commute with ghost fields, their partial derivatives and Grassmann variables $\xi,\eta,\ldots$. The ghost fields $c^a,\bar c^a$, their partial derivatives and $\xi,\eta,\ldots$ are generators of Grassmann type, that is, they anticommute with each other. Now attributing degree zero to the gauge fields $A^a_\mu$ and an auxiliary field $b^a$ (and their all partial derivatives), degree one to ghost fields (and their partial derivatives) and Grassmann variables $\xi,\eta,\ldots$, we turn the algebra of the fields of quantum Yang-Mills theory into a commutative superalgebra. An important role in the quantum theory of Yang-Mills fields is played by BRST supersymmetry, which has the form \cite{Slavnov-Faddeev}
\begin{eqnarray}
A^a_\mu\;\to\;A^a_\mu+\delta A^a_\mu,\;c^a\;\to\; c^a+\delta c^a,\;\bar c^a\; \to\;\bar c^a+\delta \bar c^a,
\end{eqnarray}
where
\begin{equation}
\delta A^a_\mu=\xi\,D_\mu c^a,\;\;\delta c^a=-\frac{1}{2}\,K^a_{bd}\,\xi\,c^b\,c^d\;,\;
    \delta\bar c^a=\xi\,b^a,
\end{equation}
and $D_\mu c^a=\partial_\mu c^a+K^a_{bd}A^b_\mu\,c^d$ is the covariant derivative.
The BRST supersymmetry can be considered as an even derivation of the commutative superalgebra of the fields of the quantum Yang-Mills theory. Indeed let us define the mapping $s$ on the basic fields as follows
\begin{eqnarray}
s A^a_\mu=(D_\mu c)^a\;,\; s c^a=-\frac{1}{2}\,K^a_{bd}\,c^b\,c^d\;,\;s\bar c^a=b^a\;,\;s b^a=0,
\end{eqnarray}
and extend this mapping to the odd derivation of the commutative superalgebra of fields by assuming that it acts on products of fields according to the graded Leibniz rule. We also assume that the derivation $s$ commutes with partial derivatives and Grassmann variables $\xi,\eta,\ldots$, i.e. $s\circ\partial_\mu=\partial_\mu\circ s, s \xi=\xi s$. It can be verified that the odd derivation $s$ has a very important property $s^2=0$. Then $\delta=\xi\,s$ is the even derivation of the commutative superalgebra of fields. Indeed for any product of basic fields $\phi,\psi$ we have
\begin{eqnarray}
\delta(\phi\,\psi) &=& \xi\,s (\phi\,\psi)=\xi\big(\delta(\phi)\,\psi+(-1)^{|\phi|}\phi\,\delta(\psi)\big)\nonumber\\
    &=& \xi\, \delta(\phi)\,\psi+(-1)^{|\phi|}\xi\,\phi\,\delta(\psi)\nonumber\\
    &=& \xi\, \delta(\phi)\,\psi+(-1)^{2|\phi|}\phi\,\xi\,\delta(\psi)=\delta(\phi)\,\psi+\phi\,\delta (\psi).\nonumber
\end{eqnarray}
Let $\Phi,\Psi$ be two elements of the commutative superalgebra of the fields of quantum Yang-Mills theory, that is, finite polynomials on the basic fields and their partial derivatives. We define the graded Lie bracket of these two polynomials by
\begin{equation}
[\Phi,\Psi]_\delta=\Phi\,\delta(\Psi)-(-1)^{|\Phi|\,|\Psi|}\,\Psi\,\delta(\Phi).
\label{graded Lie bracket for fields}
\end{equation}
Then the algebra of the fields of quantum Yang-Mills theory endowed with the graded Lie bracket (\ref{graded Lie bracket for fields}) is a transposed Poisson superalgebra.
\section{Identities in Transposed Poisson Superalgebra}
The aim of this section is to show that the identities, which hold in the case of a transposed Poisson algebra \cite{Bai-Guo-Wu-2020}, can be extended to a transposed Poisson superalgebra if we modify the identities proposed in \cite{Bai-Guo-Wu-2020}  with the help of the rule of signs.
\begin{theorem}
Let $({\mathfrak G},\cdot,[\;,\;])$ be a transposed Poisson superalgebra. Then for any $h,x,y,z\in{\mathfrak G}$ we have the following identities
\begin{eqnarray}
(-1)^{|x|\,|z|}\,x\cdot[y,z]+(-1)^{|x|\,|y|}\,y\cdot[z,x]+(-1)^{|y|\,|z|}\,z\cdot[x,y]\!\!\! &=&\!\!\! 0,\label{identity 1}\\
(-1)^{|x|\,|z|}\,[h\cdot[x,y],z]+(-1)^{|x|\,|y|}\,[h\cdot[y,z],x]+(-1)^{|y|\,|z|}\,[h\cdot[z,x],y]\!\!\! &=&\!\!\! 0,\label{identity 2}\\
(-1)^{|x|\,|z|}\,[h\cdot x,[y,z]]+(-1)^{|x|\,|y|}\,[h\cdot y,[z,x]]+(-1)^{|y|\,|z|}\,[h\cdot z,[x,y]]\!\!\! &=&\!\!\! 0,\label{identity 3}\\
(-1)^{|x|\,|z|}\,[h,x]\,[y,z]+(-1)^{|x|\,|y|}\,[h,y]\,[z,x]+(-1)^{|y|\,|z|}\,[h,z]\,[x,y]\!\!\! &=& 0,
\label{identity 4}\\
2\,u\cdot v\cdot[x,y]=(-1)^{|x|\,|v|}[u\cdot x,v\cdot y]+(-1)^{|u|\,(|x|+|v|)}[v\cdot x,u\cdot y],
\label{identity 5}\\
(-1)^{|u|\,|yv|}\,x\cdot[u,y\cdot v]+(-1)^{|v|\,|xy|}\,v\cdot[x\cdot y,u]+(-1)^{|x|\,|yv|}y\cdot[v,x]\cdot u\!\!\! &=&\!\!\! 0.
\label{identity 6}
\end{eqnarray}
\label{Theorem 2}
\end{theorem}
\begin{proof}
 In order to prove identity (\ref{identity 1}) we make use of the compatibility condition (\ref{condition of compatibility}) as follows
\begin{eqnarray}
    (-1)^{|x|\,|z|}\;2x\cdot [y,z] &=& (-1)^{|x|\,|z|}[x\cdot y,z]+(-1)^{|x|\,|z|+|x|\,|y|}[y,x\cdot z],\nonumber\\
    (-1)^{|x|\,|y|}\;2y\cdot [z,x] &=& (-1)^{|x|\,|y|}[y\cdot z,x]+(-1)^{|x|\,|y|+|y|\,|z|}[x,y\cdot x],\nonumber\\
    (-1)^{|y|\,|z|}\;2z\cdot [x,y] &=& (-1)^{|y|\,|z|}[z\cdot x,y]+(-1)^{|y|\,|z|+|x|\,|z|}[x,z\cdot y].\nonumber
\end{eqnarray}
Making use of the commutativity of a superalgebra $\mathfrak G$ and the graded skew-symmetry of bracket we can represent the right-hand sides of the above identities as follows
\begin{eqnarray}
    (-1)^{|x|\,|z|}\;2x\cdot [y,z] &=& (-1)^{|x|\,|z|}[x\cdot y,z]-(-1)^{|y|\,|z|}[z\cdot x,y],\nonumber\\
    (-1)^{|x|\,|y|}\;2y\cdot [z,x] &=& (-1)^{|x|\,|y|}[y\cdot z,x]-(-1)^{|x|\,|z|}[x\cdot y,z],\nonumber\\
    (-1)^{|y|\,|z|}\;2z\cdot [x,y] &=& (-1)^{|y|\,|z|}[z\cdot x,y]-(-1)^{|x|\,|y|}[y\cdot z,x].\nonumber
\end{eqnarray}
Summing up the left-hand sides and right-hand sides of the obtained identities, we get identity (\ref{identity 1}). The proof of identity (\ref{identity 2}) is similar to the proof of identity (\ref{identity 1}). First we use the compatibility condition (\ref{condition of compatibility}) to obtain the following three identities
\begin{eqnarray}
(-1)^{|x|\,|z|}\,2\;h\cdot[[x,y],z] &=& (-1)^{|x|\,|z|}\,[h\cdot[x,y],z]+(-1)^{|x|\,|z|+(xy,h)}[[x,y],h\cdot z],\nonumber\\
(-1)^{|x|\,|y|}\,2\;h\cdot[[y,z],x] &=& (-1)^{|x|\,|y|}\,[h\cdot[y,z],x]+(-1)^{|x|\,|y|+(yz,h)}[[y,z],h\cdot x],\nonumber\\
(-1)^{|y|\,|z|}\,2\;h\cdot[[z,x],y] &=& (-1)^{|y|\,|z|}\,[h\cdot[z,x],y]+(-1)^{|y|\,|z|+(xz,h)}[[z,x],h\cdot y],\nonumber
\end{eqnarray}
where $(xy,h)=|x|\,|h|+|y|\,|h|$,$(yz,h)=|y|\,|h|+|z|\,|h|$,$(xz,h)=|x|\,|h|+|z|\,|h|$.
Due to the super Jacobi identity the sum of left-hand sides of the above identities is equal to zero. Hence we get
\begin{eqnarray}
&& (-1)^{|x|\,|z|}\,[h\cdot[x,y],z]+(-1)^{|x|\,|y|}\,[h\cdot[y,z],x]+(-1)^{|y|\,|z|}\,[h\cdot[z,x],y]\nonumber\\
      &&\quad\quad\quad\qquad +(-1)^{|x|\,|z|+(xy,h)}[[x,y],h\cdot z]+(-1)^{|x|\,|y|+(yz,h)}[[y,z],h\cdot x]\nonumber\\
      &&\qquad\qquad\qquad\qquad\qquad\qquad\qquad\qquad\qquad+(-1)^{|y|\,|z|+(xz,h)}[[z,x],h\cdot y]=0.\label{relation 1}
\end{eqnarray}
The sum of the first three terms in the above identity is the left-hand side of identity (\ref{identity 2}). Thus in order to prove identity (\ref{identity 2}) we need to show that the sum of the last three terms in the above identity is zero. For the first one of them, that is, $[[x,y],h\cdot z]$ (we temporarily omit the coefficient $(-1)^{|x|\,|z|+(xy,h)}$ , which will be taken into account later) we have the super Jacobi identity
\begin{equation}
(-1)^{|x|\,|h|+|x|\,|z|}[[x,y],h\cdot z]+(-1)^{|x|\,|y|}[[y,h\cdot z],x]+(-1)^{|y|\,|h|+|y|\,|z|}[[h\cdot z,x],y]=0.
\label{super Jacobi identity}
\end{equation}
Applying to the middle term at the left-hand side of this equality the compatibility condition
$$
2\,h\cdot [y,z]=[h\cdot y,z]+(-1)^{|y|\,|h|}[y,h\cdot z]\;\;\Rightarrow\;\;[y,h\cdot z]=(-1)^{|y|\,|h|}(2\,h\cdot [y,z]-[h\cdot y,z]),
$$
and substituting the obtained expression multiplied by $(-1)^{|x|\,|z|+(xy,h)}$ into super Jacobi identity (\ref{super Jacobi identity}) we get
\begin{eqnarray}
(-1)^{|y|\,|h|}[[x,y],h\cdot z]+(-1)^{(yzh,x)}([2\,h\cdot [y,z],x]\!\!\! \!&-&\!\!\!\! [[h\cdot y,z],x])\nonumber\\
    &+&\!\!\!\!\!(-1)^{|y|\,|z|+(zh,x)}[[h\cdot z,x],y]=0,\label{relation 2}
\end{eqnarray}
where $(yzh,x)=|y|\,|x|+|z|\,|x|+|h|\,|x|$ and $(zh,x)=|z|\,|x|+|h|\,|x|$. Applying the same calculations to the next two terms in (\ref{relation 1}), we get two more identities
\begin{eqnarray}
(-1)^{|z|\,|h|}[[y,z],h\cdot x]+(-1)^{(xzh,y)}([2\,h\cdot [z,x],y]\!\!\! \!&-&\!\!\!\! [[h\cdot z,x],y])\nonumber\\
    &+&\!\!\!\!\!(-1)^{|x|\,|z|+(xh,y)}[[h\cdot x,y],z]=0,\label{relation 3}\\
(-1)^{|x|\,|h|}[[y,z],h\cdot x]+(-1)^{(xyh,z)}([2\,h\cdot [z,x],y]\!\!\! \!&-&\!\!\!\! [[h\cdot z,x],y])\nonumber\\
    &+&\!\!\!\!\!(-1)^{|x|\,|y|+(yh,z)}[[h\cdot x,y],z]=0,\label{relation 4}
\end{eqnarray}
Summing up the identities (\ref{relation 2}), (\ref{relation 3}), (\ref{relation 4}) multiplied by $(-1)^{|x|\,(|z|+|h|)},(-1)^{|y|\,(|x|+|h|)}$ and $(-1)^{|z|\,(|y|+|h|)}$ respectively, we obtain
\begin{eqnarray}
&&(-1)^{|x|\,|z|+(xy,h)}[[x,y],h\cdot z]+(-1)^{|x|\,|y|+(yz,h)}[[y,z],h\cdot x]+(-1)^{|y|\,|z|+(xz,h)}[[z,x],h\cdot y]\nonumber\\
  &&2\,\big((-1)^{|x|\,|z|}\,[h\cdot[x,y],z]+(-1)^{|x|\,|y|}\,[h\cdot[y,z],x]+(-1)^{|y|\,|z|}\,[h\cdot[z,x],y]\big)=0.
\label{relation 5}
\end{eqnarray}
Subtracting the identity (\ref{relation 1}) from (\ref{relation 5}) we get the identity (\ref{identity 2}). By virtue of the proven identity (\ref{identity 2}), the sum of the first three terms in (\ref{relation 1}) will be equal to zero. The remaining three terms, after a suitable permutation of the arguments, give the identity (\ref{identity 3}).

In order to prove the last identity (\ref{identity 4}) we apply the compatibility condition to the following products
\begin{eqnarray}
2\,[x,y]\cdot[h,z] &=& [[x,y]\cdot h,z]+(-1)^{|h|(|x|+|y|)}[h,[x,y]\cdot z],\nonumber\\
2\,[y,z]\cdot[h,x] &=& [[y,z]\cdot h,x]+(-1)^{|h|(|y|+|z|)}[h,[y,z]\cdot x],\nonumber\\
2\,[z,x]\cdot[h,y] &=& [[z,x]\cdot h,y]+(-1)^{|h|(|x|+|z|)}[h,[z,x]\cdot y].\nonumber
\end{eqnarray}
In the first brackets on the right-hand sides of these relations, we rearrange $h$ and $[x,y]$, $h$ and $[y,z]$, $h$ and $[z,x]$. Next we multiply the first relation by $(-1)^{|x|\,|h|+|y|\,|h|+|x|\,|z|}$, the second by $(-1)^{|y|\,|h|+|z|\,|h|+|x|\,|y|}$ and the third by $(-1)^{|x|\,|h|+|z|\,|h|+|y|\,|z|}$. Finally we rearrange the brackets at the left-hand side of every relation. We obtain
\begin{eqnarray}
(-1)^{|y|\,|z|}\,2\,[h,z]\cdot[x,y] &=& (-1)^{|x|\,|z|}[h\cdot [x,y],z]+[h, (-1)^{|x|\,|z|})[x,y]\cdot z],\nonumber\\
(-1)^{|x|\,|z|}\,2\,[h,x]\cdot[y,z] &=& (-1)^{|x|\,|y|}[h\cdot [y,z],x]+[h,(-1)^{|x|\,|y|}[y,z]\cdot x],\nonumber\\
(-1)^{|x|\,|y|}2\,[h,y]\cdot[z,x] &=& (-1)^{|y|\,|z|}[h\cdot [z,x],y]+[h,(-1)^{|y|\,|z|}[z,x]\cdot y].\nonumber
\end{eqnarray}
Summing up the left-hand sides and the right-hand sides of these relations we get the relation, whose left-hand side (multiplied by 1/2) coincides with the left-hand side of the identity (\ref{identity 4}), but the right-hand side is equal to zero. Indeed the sum of the first terms at the right-hand side of every relation vanishes because of the identity (\ref{identity 2}). The sum of the second terms vanishes because of the identity (\ref{identity 1}).  The identity (\ref{identity 5}) is proved by means of the compatibility condition. In order to prove the identity (\ref{identity 6}) we take the sum of the three following relations: the compatibility condition for the products $2\,x\cdot [u,y\cdot v]$, $2\,v\cdot [x\cdot y,u]$ and the identity (\ref{identity 5}) multiplied by 2.  This ends the proof of Theorem \ref{Theorem 2}.
\end{proof}

\section{CONCLUSIONS}
The impetus for writing this article was the paper \cite{Bai-Guo-Wu-2020}, which introduces the concept of a transposed Poisson algebra, describes its structure, gives important examples and relations to other algebraic structures. Our goal in this article is to show that the notion of a transposed Poisson algebra can be extended to superalgebras in substance merely by means of Manin's rule  of signs, and also to show that important structures of a transposed Poisson algebra (example constructed with the help of derivation and identities) also hold in a transposed Poisson superalgebra in a graded form. We think that the example of transposed Poisson superalgebra constructed by means of differential forms and Lie derivative reveals a geometric meaning of the notion of transposed Poisson superalgebra. In the paper \cite{Bai-Guo-Wu-2020} the authors show that the identities in a transposed Poisson algebra play an important role in establishing a relation with Hom-Lie algebras. We hope that the graded versions of these identities proved in this paper play the same important role in establishing connections between a transposed Poisson superalgebra and Hom-Lie superalgebras. This, and a notion of a transposed Poisson 3-Lie superalgebra, is our next goal in developing this direction.

\section*{ACKNOWLEDGEMENTS}
The authors express their deep gratitude to the colleagues and students participating in the Seminar of Geometry and Topology at the Institute of Mathematics and Statistics of the University of Tartu for useful discussions of the structures presented in this article.






\end{document}